\documentclass[letter, 10pt, conference]{ieeeconf}  
\IEEEoverridecommandlockouts 

\usepackage{graphicx}      
\usepackage{xcolor}
\usepackage{amsmath}
\usepackage{amssymb}
\usepackage{subcaption}
\usepackage{algorithm}
\usepackage{algpseudocode}
\usepackage{algorithm}
\usepackage{epstopdf}
\usepackage{placeins}
\usepackage{hyperref}
\usepackage{tikz}
\usepackage{enumerate}
\usetikzlibrary{calc}
\newtheorem{assumption}{Assumption}
\newtheorem{proposition}{Proposition}

\title{\LARGE \bf
	Active exploration in adaptive model predictive control*
}

\author{Anilkumar Parsi, Andrea Iannelli and Roy S. Smith%
\thanks{* This work is supported by the Swiss National Science Foundation under grant number 200021\_178890. The authors are with the Automatic Control Laboratory, ETH Zurich, Switzerland. %
 {\tt\small \{aparsi,iannelli,rsmith\}@control.ee.ethz.ch }%
}
}

\begin{document}
		
\maketitle

\begin{abstract}                
A dual adaptive model predictive control (MPC) algorithm is presented for linear, time-invariant systems subject to bounded disturbances and parametric uncertainty in the state-space matrices. Online set-membership identification is performed to reduce the uncertainty and thus control affects both the informativity of identification and the system's performance. The main contribution of the paper is to include this dual effect in the MPC optimization problem using a predicted worst-case cost in the objective function. This allows the controller to perform active exploration, that is, the control input reduces the uncertainty in the regions of the parameter space that have most influence on the performance. Additionally, the MPC algorithm ensures robust constraint satisfaction of state and input constraints. Advantages of the proposed algorithm are shown by comparing it to a passive adaptive MPC algorithm from the literature.
\end{abstract}		
\section{Introduction} 
Adaptive control is a technique where the controller parameters are updated using measurement data. Conventional methods of adaptive control like gain scheduling are based on certainty equivalence, and neglect the model uncertainties \cite{sastry2011adaptive}. For this reason, they cannot handle constraints on states and inputs of the system. Model predictive control (MPC) is a popular technique since it guarantees stability and constraint satisfaction under uncertainty \cite{rawlings2009model}. The structure of MPC controllers facilitates easy integration of model adaptation into the controller. Utilizing this advantage, a variety of adaptive MPC schemes have been proposed in the recent past using different model structures (state-space, impulse response, etc.) and adaptation methods (set-membership identification, recursive least-squares, etc.)  \cite{kim2008adaptive,tanaskovic2014adaptive,lorenzen2017,heirung2017dual}.

An adaptive MPC algorithm using set-membership identification is proposed in \cite{tanaskovic2014adaptive}, which ensures robust constraint satisfaction for systems affected by bounded measurement noise. An extension of this algorithm was proposed in \cite{parsi2019robust}, where a worst-case cost is used to improve robustness of the performance. The algorithm uses an impulse response model which depends on a large number of parameters. An alternative method has been proposed in \cite{lorenzen2017} which uses uncertain state-space models subject to bounded disturbances. Using tube-MPC, the algorithm ensures robust constraint satisfaction while reducing the uncertainty online. However, in all these methods the adaptation is passive, that is, the MPC optimizer does not exploit the fact that identification and control are being simultaneously performed.

These disadvantages can be addressed using dual control \cite{Wittenmark95adaptivedual}, a technique which computes control inputs under decision relevant, reducible uncertainty. An optimal dual control problem can be formulated by  modeling the dependence of uncertainty reduction on the control inputs. The solution to this problem is given by dynamic programming, whose computational complexity is high \cite{mesbah2018stochastic}. Instead, the existing dual control algorithms approximate the optimal control problem using heuristics to add a probing effect on the control input. In \cite{lorenzen2019}, an adaptive MPC algorithm is presented with a constraint on the control input to ensure persistent excitation. Though parameter convergence is guaranteed, this method could result in excessive probing, especially after the uncertainty is reduced.   In \cite{weiss2014robust}, the parameter error covariance is penalized in the objective function of MPC. However, a lower covariance of the parameter uncertainty might not always translate into improved performance, and the cost function requires tuning the probing effect. These problems can be mitigated by using an application-oriented approach to dual control \cite{larsson2016application},  \cite{Iannelli2020}. Here the probing effect is induced by using a measure of the robust performance, such as worst-case cost, instead of geometric measures of uncertainty. This ensures that active exploration is performed, that is, the uncertainty is reduced in regions of parameter space to improve the control performance.

The main contribution of this paper is to formulate active exploration in a dual adaptive MPC framework. For this purpose, the regulation of a linear, time-invariant system with affine uncertainty in the state space matrices is considered. The system is subject to bounded disturbances and must satisfy state and input constraints. Using an approach similar to \cite{aswani2013provably}, the problems of feasibility and learning are decoupled by using two state tubes. A robust state tube is used to ensure feasibility for all model parameters in an initial parameter set, which can be large. A predicted parameter set is then defined as a function of the control input, and a predicted state tube is constructed such that it is robust to uncertainties in predicted parameter set. The cost function is defined as the worst-case cost over the predicted state tube. The algorithm requires a non-convex optimization problem to be solved online. The algorithm has the flexibility to trade-off between the computational complexity and quality of active exploration using the length of the predicted state tube. The performance of the algorithm with varying predicted state tube lengths is compared to a passive adaptive MPC algorithm using numerical simulations. 


\subsection{Notation}
The sets of real numbers and non-negative real numbers are denoted by $ \mathbb{R} $ and $ \mathbb{R}_{\ge0} $ respectively. The sequence of integers from $ n_1 $ to $ n_2 $ is represented by $ \mathbb{N}_{n_1}^{n_2} $. For a vector $ b $, $ b^{\intercal} $ represents its transpose, and $ [b]_i $ refers to its $ i ^{th}$ element. The $ i ^{th}$ row of a matrix $ A $ is denoted by $ [A]_{i} $. The dimensions of matrices and vectors are not explicitly specified when they can be inferred from the context. For any real scalar-valued function $ J $, $ \displaystyle\max_{h \in \mathbb{H}} J(h)$ refers to the maximum value of $ J $ over the set $ \mathbb{H} $. The Minkowski sum of two sets $ A$ and $ B $ is denoted by $ A \oplus B $, and $ \mathbf{1} $ denotes a column vector of 
appropriate length whose elements are equal to 1.  The convex hull of the elements of a set S is represented by co\{S\}. The notation $ a_{l|k} $ denotes the value of $ a $ at time step $ k+l $ predicted at the time step $ k $.

\section{Problem configuration}
\subsection{System description}
We consider a discrete time, linear time-invariant system with state $x_k \in \mathbb{R}^n$, control input $u_k \in \mathbb{R}^m$ and disturbance $w_k \in \mathbb{W} \subset \mathbb{R}^n$ at the time step $k$. The system dynamics can be described according to the parametric equation
\begin{equation} \label{eq:Dynamics}
    x_{k+1} = A(\theta) x_k + B(\theta) u_k + w_k,
\end{equation}
where $\theta \in \mathbb{R}^p$ is an unknown, constant parameter and $\theta^*$ is its true value. It is assumed that all the state variables are measurable. The state matrices are parameterized as
\begin{align} \label{eq:Parameterization}
\begin{split}
    A(\theta) = A_0 + \displaystyle\sum_{i=1}^{p} A_i [\theta]_i, \quad  B(\theta) = B_0 + \displaystyle\sum_{i=1}^{p} B_i [\theta]_i,
\end{split}
\end{align}
and  $\theta$ belongs to the bounded polytope
\begin{equation}\label{eq:ParameterBounds}
    \Theta := \{\theta \in \mathbb{R}^p | H_{\theta} \theta \le h_{\theta} \},
\end{equation}
such that $ \theta^* \in \Theta$ and $H_{\theta} \in \mathbb{R}^{n_\theta \times p}  $. The states and inputs of the system must satisfy the constraints
\begin{equation} \label{eq:Constraints}
\mathbb{Z} = \left\{(x_k,u_k) \in \mathbb{R}^n \times \mathbb{R}^m \bigr|  F x_k + G u_k \le \mathbf{1}\right\},
\end{equation} 
where $ \mathbb{Z} $ is a compact set and $ F \in \mathbb{R}^{n_c \times n} $. The objective is to regulate the system state from the initial condition $x_0$ to the origin, while robustly satisfying the constraints in \eqref{eq:Constraints}. 
\begin{assumption}\label{As:Disturbance}
The disturbance set $\mathbb{W}$ is a bounded polytope described by the $ n_w $ constraints in the set
\begin{equation}
    \mathbb{W} = \{w \in \mathbb{R}^n | H_w w \le h_w \}.
\end{equation}
\end{assumption}

\subsection{Online set-membership identification}
Set-membership identification is a technique used to identify systems affected by bounded noise with unknown statistical properties \cite{milanese1991}. The identification procedure defines a feasible parameter set (FPS), which contains the set of all parameters to be robust against. The FPS is initialized with $ \Theta $ and updated at each time step $ k $ to $ \Theta_k $. To perform the update, a set of non-falsified parameters is constructed using measurement data from the previous $ s$ time steps as

\begin{align} \label{eq:SimpleNonfalsified}
\begin{split}
\Delta_{k} &:= \biggr\{
\theta \in \mathbb{R}^{p}\: \biggr|x_{t+1} {-} A(\theta) x_t {-} B(\theta) u_t \in \mathbb{W},\: \forall t \in \mathbb{N}_{k-s}^{k-1}
\biggr\}\\ 
&=\biggr\{
\theta \in \mathbb{R}^{p}\: \biggr| \: -H_w D_t \theta  \le h_w+H_w d_{t+1}, \forall t \in \mathbb{N}_{k-s}^{k-1}
\biggr\}\\
&= \left\{\theta \in \mathbb{R}^{p}\: |\: H_\Delta \theta \le h_\Delta \right\},
\end{split}
\end{align}
where $D_t\in\mathbb{R}^{n\times p}$ and $d_{k+1}\in \mathbb{R}^{n}$ are 
\setlength\arraycolsep{2pt}
\begin{align}\label{eq:Dtdt}
\begin{split}
    D_t &:= D(x_t,u_t) = \begin{bmatrix}
         A_1 x_t {+} B_1 u_t, & \ldots, & A_p x_t {+} B_p u_t
    \end{bmatrix},\\
    d_{t+1} &:= A_0 x_t + B_0 u_t - x_{t+1}, \qquad \forall t \in \mathbb{N}_{k-s}^{k-1}.
\end{split}
\end{align}
Note that $ D_t $ and $ d_t $ are quantities which linearly depend on the measured state and input vectors, but the dependence is omitted for clarity. This notation is adopted so that $ \Delta $ can be represented by hyperplane constraints in $ \mathbb{R}^p $.

In \eqref{eq:SimpleNonfalsified}, the non-falsified set $\Delta_{k}$ defines the set of all parameters that could have generated the measurement sequence $ \{x_{k-s},\ldots,x_{k}\}$. The set $ \Theta_k $ is defined using a fixed number of polytopic constraints given by
\begin{equation} \label{eq:Theta_k_def}
\Theta_k :=  \{\theta \in \mathbb{R}^p| H_{\theta} \theta\le h_{\theta_k}\}.
\end{equation}
The matrix $ H_\theta $ is chosen offline and $ h_{\theta_k}$ is updated online such that 
\begin{equation}\label{eq:Theta_k_update}
 \Theta_k \supseteq \Theta_{k-1} \cap \Delta_{k} 
\end{equation}
is satisfied. This is ensured by calculating $  h_{\theta_k} $ as a solution to the following set of linear programs:
\begin{align}\label{eq:Theta_k_LP}
\begin{split}
[h_{\theta_k}]_i \: =\: &\max_{\theta\in \mathbb{R}^{p}} \quad  [H_\theta]_{i} \theta\\
 &\text{s. t. } \quad  
 \begin{bmatrix}
    H_\theta \\ H_\Delta
\end{bmatrix} \theta
 \le \begin{bmatrix}
  h_{\theta_{k-1}} \\ h_{\Delta}
 \end{bmatrix}, \quad i = 1,2,\ldots,n_\theta .\\
\end{split}
\end{align}

\section{Robust state tube and constraints}
\subsection{Tube MPC}
To ensure robust constraint satisfaction, the tube MPC approach proposed in \cite{langson2004} is used. The prediction horizon of the MPC problem is $ N $, and the control input is parameterized using a feedback gain $ K $ as 
\begin{equation}\label{eq:InputParameterization}
u_{l|k} = Kx_{l|k} + v_{l|k},
\end{equation}
where $ \{v_{l|k}\}_{l=0}^{N-1} $ are decision variables in the MPC optimization problem.
\begin{assumption}\label{As:Feedback}
The feedback gain $ K $ is chosen such that $ A_{\text{cl}}(\theta)  = A( \theta) + B(\theta)K$ is asymptotically stable $ \forall \theta \in \Theta $. 
\end{assumption}
The gain $ K $ can be computed using standard robust control techniques, for example, following the approach in \cite{kothare1996}.

\noindent A state tube is defined using the set-based dynamics
\begin{subequations}\label{eq:SetDynamics}
\begin{align}
\mathbb{X}_{0|k} &\ni \{x_{k}\},  \label{eq:initialSet}\\
\mathbb{X}_{l+1|k} &\supseteq A(\theta)\mathbb{X}_{l|k}  \oplus B(\theta)u_{l|k} \oplus \mathbb{W} , \label{eq:SetInclusions}\\
& \qquad  \forall \theta \in \Theta_k, \quad l = 0,1,\ldots,N-1, \nonumber
\end{align}
\end{subequations}
which ensures that $ x_{l|k} \in \mathbb{X}_{l|k} $ for all the realizations of uncertainty and disturbance. The tube cross-section at each time step, $ \mathbb{X}_{l|k} $, is parameterized by translation and scaling of the set 
\begin{align}
 \mathbb{X}_0 := \{x| H_x x \le \mathbf{1}\} = \text{co}\{x^{1},x^{2},\ldots,x^{v}\} ,
\end{align}
where the vertices $\{x^{1},x^{2},\ldots,x^{v}\} $  and the matrix $ H_x \in \mathbb{R}^{n_x\times n} $ are computed offline. The variables $ z_{l|k} \in \mathbb{R}^{n}$ and $ \alpha_{l|k} \in \mathbb{R}_{\ge0}$ define the translation and scaling of $ \mathbb{X}_{0} $ respectively, and are decision variables in the MPC optimization. Then, for $ l = \mathbb{N}_{0}^{N}, $ the state tube is parameterized as
\begin{align}\label{eq:StateTubeParameterization}
\begin{split}
\mathbb{X}_{l|k} &= \{z_{l|k}\} \oplus \alpha_{l|k} \mathbb{X}_{0} \quad = \{x| H_x (x-z_{l|k}) \le \alpha_{l|k} \mathbf{1}\}\\
&= \{z_{l|k}\} \oplus \alpha_{l|k} \text{co}\{x^{1},x^{2},\ldots,x^{v}\}.
\end{split}
\end{align}

\subsection{Reformulation of constraints}
The state and input constraints defined in \eqref{eq:Constraints} and the set dynamics proposed in \eqref{eq:SetDynamics} must be robustly satisfied for all $ \theta \in \Theta_k $ and disturbances in $\mathbb{W}$. To reformulate these in a convex manner, the following notation is defined
\begin{equation}\label{eq:x_jlk}
\begin{array}{rll}
x_{l|k}^{j} &= z_{l|k} + \alpha_{l|k} x^{j}, \quad &d_{l|k}^{j} = A_0 x_{l|k}^{j} + B_0 u_{l|k}^{j} - z_{l+1|k},  \\
u_{l|k}^{j} &= Kx_{l|k}^{j} + v_{l|k}, \quad & D_{l|k}^{j} = D(x_{l|k}^{j},u_{l|k}^{j}),\\
\end{array}
\end{equation}
where $  j \in \mathbb{N}_{1}^{v} ,l \in \mathbb{N}_{0}^{N-1} $. Note that unlike the definition in \eqref{eq:Dtdt} where $ D_t,d_t $ are a function of known states and inputs, the quantities $ D_{l|k}^{j}, d_{l|k}^{j} $ linearly depend on the decision variables of MPC.  Additionally, the vectors $ \bar{f}$ and $\bar{w} $ are computed offline such that for $  i \in \mathbb{N}_{1}^{n_c},  j \in \mathbb{N}_{1}^{n_x}  $
\begin{equation}
\begin{split}
 [\bar{f}]_{i} &= \displaystyle\max_{x\in \mathbb{X}_0} [F+GK]_i x,\quad [\bar{w}]_{j} = \displaystyle\max_{w\in \mathbb{W}} \: [H_x]_j w.
\end{split}
\end{equation}
The following proposition from \cite{lorenzen2017} reformulates the robust constraints and set-dynamics  as linear equality and inequality constraints.
\begin{proposition}\label{Pr:SetDynamics}
	Let the state tube $ \{\mathbb{X}_{l|k}\}_{l=0}^{N} $ be parameterized according to \eqref{eq:StateTubeParameterization}. Then, the constraints \eqref{eq:Constraints} and set-dynamics \eqref{eq:SetDynamics} are satisfied if and only if $\forall  j{\in} \mathbb{N}_{1}^{v} $, $ l{\in} \mathbb{N}_{0}^{N-1}$ there exists $ \Lambda_{l|k}^{j} \in \mathbb{R}^{n_x\times n_\theta}_{\ge0}$ such that 
	\begin{subequations}\label{eq:lambdaConstraints}
		\begin{align}
		(F+GK)z_{l|k} + Gv_{l|k} + \alpha_{l|k}\bar{f} &\le \mathbf{1},\\
		-H_x z_{0|k} -\alpha_{0|k}\mathbf{1} &\le -H_x x_k ,		\label{eq:x0Constraint}\\
		\Lambda_{l|k}^{j} h_{\theta_k} + H_x d_{l|k}^{j} -\alpha_{l+1|k} \mathbf{1} &\le -\bar{w} ,\label{eq:InclusionIneq}\\
		H_x D_{l|k}^{j} &= \Lambda_{l|k}^{j} H_{\theta} \label{eq:InclusionEqual}.
		\end{align} 
	\end{subequations}
\end{proposition}
\subsection{Terminal set}
To obtain an MPC algorithm which ensures recursive feasibility, the state tube is directed to a terminal set. The terminal constraints are imposed on $ z_{N|k} $ and $\alpha_{N|k} $ since they define the last cross section of the state tube. 
\begin{assumption}\label{As:TerminalSet}
	There exists a nonempty terminal set $  \mathbb{X}_{T} = \{(z,\alpha)\in \mathbb{R}^n{\times}\mathbb{R} | \: z {=} 0,\: \alpha {\in} [0,\bar{\alpha}] \}$, such that  for all $ \theta\in \Theta $ it holds that
	\begin{align*}
		\alpha \in [0,\bar{\alpha}] \implies &\exists \alpha^{+}\in [0,\bar{\alpha}] \: \text{s.t.}\\
		& A_{\text{cl}}(\theta) (\alpha\mathbb{X}_0) \oplus \mathbb{W}\subseteq  \alpha^+\mathbb{X}_0,\\
		\alpha \in [0,\bar{\alpha}] \implies& (x,Kx)\in \mathbb{Z} \quad \forall x\in \alpha\mathbb{X}_0.
	\end{align*}  
\end{assumption}
Assumption \ref{As:TerminalSet} implies that the set $ \mathbb{X}_{T} $ is a robust positively invariant (RPI) set for the set-dynamics in $ (z,\alpha) $, with an additional constraint that the set $ \mathbb{X}_{N|k} $ remains centered at origin. Note that Assumption \ref{As:Feedback} is a necessary condition for Assumption \ref{As:TerminalSet} to be satisfied, but they are stated separately to emphasize that the stronger assumption is only needed to implement the terminal condition. Thus, the terminal constraint for the MPC algorithm is 
\begin{equation}\label{eq:TerminalConstraint}
z_{N|k} =0 ; \quad  \alpha_{N|k} \le \bar{\alpha} .
\end{equation}

\section{Predicted state tube for exploration}
In this section predicted variables and sets are defined, which are analogous to the variables and sets defined in the previous section. Each of the predicted quantities is denoted with a hat ($ \hat{~} $) on the top. Using the predicted sets and variables, the dual effect of the input is captured by the MPC optimization problem. The control input $ u_k $ and the future parameter set $ \Theta_{k+1} $ are connected by the identification step \eqref{eq:Theta_k_update}.  By predicting the next state measurement $ \hat{x}_{1|k} $, a predicted parameter set $\hat{\Theta}_k$ is defined as a function of $ u_k $. A predicted state tube is then constructed to contain all the state trajectories generated by the predicted parameter set. 

\subsection{Predicted parameter set} 
To predict the next measurement, an estimate of the parameters is required. For this purpose, a least mean squares filter is used to calculate $\hat{\theta}_k$ as an estimate of $ \theta^* $ \cite{lorenzen2019}.  Alternative filters such as recursive least squares can also be used for this purpose. Using  $\hat{\theta}_k$, the predicted state measurement at the next time step can be written as
\begin{equation}\label{eq:xhat}
\hat{x}_{1|k} = A(\hat{\theta}_k) x_k + B(\hat{\theta}_k)u_k.
\end{equation}
Using the predicted state $ \hat{x}_{1|k} $, the future constraints on the parameter set $ \Theta_k $ are given by
\begin{align}
\begin{split}
\hat{\Delta}_k &:= \{\theta \in \mathbb{R}^{p}| \hat{x}_{1|k} {-} A(\theta) x_k {-} B(\theta) u_k \in \mathbb{W}\}, \\
				 &= \{\theta \in \mathbb{R}^{p}|-H_w D_k\theta \le h_w {-}H_w D_k \hat{\theta}_k   \}, \\
\end{split}
\end{align}
where $ u_k = K x_k + v_{0|k} $ is the first control input calculated by the MPC controller. Since $ u_k $ is the only input applied in closed loop, the predicted constraints from other inputs are not considered. A predicted parameter set $ \hat{\Theta}_k \subseteq \Theta_k $ can now be defined as
\begin{align}\label{eq:PredictedParameterSet}
\begin{split}
\hat{\Theta}_k := \Theta_k \cap \hat{\Delta}_k 
				&= \left\{\theta \in \mathbb{R}^{p} \biggr| 
					\begin{array}{rl}
					H_{\theta} \theta &\le h_{\theta_k} \\
					-H_w D_k \theta &\le  h_w -H_w D_k \hat{\theta}_k 
					\end{array} \right\} \\
			    &= \: \{\theta \in \mathbb{R}^{p} | \hat{H}_{\theta} \theta \le \hat{h}_{\theta_k} \}.
\end{split}
\end{align}
Note that $ \hat{\Theta}_k $ is dependent on the control input through the definition of $ D_k $ \eqref{eq:Dtdt}, but this dependence is omitted for clarity. This captures the effect of the control input on the identification. Figure \ref{fig:parameterSetExample} shows the parameter estimate $ \hat{\theta}_k $, the parameter sets $ \Theta_k$ and $ \hat{\Theta}_k$, and the predicted constraints. 
{\setlength\belowcaptionskip{-1ex}
\begin{figure*}
	\centering
	\begin{subfigure}[t]{.38\textwidth}
		\centering
		\begin{tikzpicture}
		
		\draw [blue,thick](0,0) rectangle (2.5,2.5);
		\node[anchor=north,blue] at (0,0){$ \Theta_k $};
		
		\draw [dashed,red,thick](0.3,3.2) -- (3.0,0.5);
		\draw [dashed,red,thick](-0.8,2.3) -- (2.0,-0.5);
		\node[anchor=north,red] at (2.5,0){$ \hat{\Theta}_k $};
		
		\fill[red!20!white] (1.5,0) -- (2.5,0)-- (2.5,1) -- (1,2.5) -- (0,2.5) -- (0,1.5) --cycle;
		
		\coordinate (x0) at (1.7,0.8);
		\filldraw (x0) circle (0.5pt) node[anchor=east] {$ \hat{\theta}_k $};
		\end{tikzpicture}
		\caption{The estimated parameter is $ \hat{\theta}_k $. The parameter set $ \Theta_k $ is bounded by the blue constraints and the dashed lines represent the predicted constraints. The shaded region shows the predicted parameter set $ \hat{\Theta}_k $.}
		\label{fig:parameterSetExample}
	\end{subfigure}\hfill
	\begin{subfigure}[t]{.59\textwidth}
		\begin{tikzpicture}
		\coordinate (x0) at (-8.0,0.5);
		\filldraw (x0) circle (0.5pt) node[anchor=east] {$ x_k $};
		
		\def \a1{0.4}
		\coordinate (a) at (-7.0,1);
		\def \ad1{0.31}
		\coordinate (ad) at (-7.05,0.95);
		\draw[blue,thick] ($ (a) + \a1*(0,-1) $) -- ($ (a) +\a1*(-.95,-.31) $) -- ($ (a) +\a1*(-.59,.81) $) -- ($ (a) +\a1*(.59,.81) $) -- ($ (a) +\a1*(.95,-.31) $) -- cycle;
		\draw[dashed,red,thick] ($ (ad) + \ad1*(0,-1) $) -- ($ (ad) +\ad1*(-.95,-.31) $) -- ($ (ad) +\ad1*(-.59,.81) $) -- ($ (ad) +\ad1*(.59,.81) $) -- ($ (ad) +\ad1*(.95,-.31) $) -- cycle;
		\node[anchor=south,red] at ($ (a) + \a1*(0,.81) $){$ \mathbb{\hat{X}}_{1|k} $};
		\node[anchor=north,blue] at ($ (a) + \a1*(0,-1) $){$ \mathbb{X}_{1|k} $};
		\filldraw[red] (ad) circle (0.5pt);
		
		\def\b1{0.8}
		\coordinate (b) at (-5.,1.25);
		\def \bd1{0.6}
		\coordinate (bd) at (-5.1,1.2);
		\draw[blue,thick] ($ (b) +\b1*(0,-1) $) -- ($ (b) +\b1*(-.95,-.31) $) -- ($ (b) +\b1*(-.59,.81) $) -- ($ (b) +\b1*(.59,.81) $) -- ($ (b) +\b1*(.95,-.31) $) -- cycle;
		\draw[dashed,red,thick] ($ (bd) + \bd1*(0,-1) $) -- ($ (bd) + \bd1*(-.95,-.31) $) -- ($ (bd) + \bd1*(-.59,.81) $) -- ($ (bd) + \bd1*(.59,.81) $) -- ($ (bd) + \bd1*(.95,-.31) $) -- cycle;
		\node[anchor=south,red] at ($ (b) + \b1*(0,.81) $){$ \mathbb{\hat{X}}_{2|k} $};
		\node[anchor=north,blue] at ($ (b) + \b1*(0,-1) $){$ \mathbb{X}_{2|k} $};
		\filldraw[red] (bd) circle (0.5pt) node[anchor=north,red] {$ \hat{z}_{2|k} $};
		
		\def\c1{0.9}
		\coordinate (c) at (-2.5,0.9);
		\draw[blue,thick] ($ (c) +\c1*(0,-1) $) -- ($ (c) +\c1*(-.95,-.31) $) -- ($ (c) +\c1*(-.59,.81) $) -- ($ (c) +\c1*(.59,.81) $) -- ($ (c) +\c1*(.95,-.31) $) -- cycle;
		\node[anchor=north,blue] at ($ (c) + \c1*(0,-1) $){$ \mathbb{X}_{3|k} $};
		\filldraw[blue] (c) circle (0.5pt) node[anchor=north,blue] {$ {z}_{3|k} $};
		
		\def\alphaT{1.7}
		\draw[blue,thick] (0,-1) -- (-.95,-.31) -- (-.59,.81) -- (.59,.81) -- (.95,-.31) -- cycle;
		\draw[ultra thick] ($ \alphaT*(0,-1) $) -- ($ \alphaT*(-.95,-.31) $) -- ($ \alphaT*(-.59,.81) $) -- ($ \alphaT*(.59,.81) $) -- ($ \alphaT*(.95,-.31) $)  -- cycle;
		\node[anchor=north,blue] at (0,-0.95){$ \mathbb{X}_{4|k} $};
		\node[anchor=east,black] at ($ \alphaT*(0,-1) $){$ \mathbb{X}_{T} $};
		\filldraw (0,0) circle (0.5pt) node[anchor=north] {O};
		\end{tikzpicture}
		\caption{The state tube $ \{\mathbb{X}_{l|k}\}_{l=1}^{4} $ is shown in blue and predicted state tube $ \{\mathbb{\hat{X}}_{l|k}\}_{l=1}^{2} $ is shown in red (dashed). The values of $ N $ and $ \hat{N} $ are 4 and 2 respectively. The depiction of the terminal set $ \mathbb{X}_T $ in $ \mathbb{R}^n $ is shown in black and contains the set $ \mathbb{X}_{N|k} $ centered at origin.}
		\label{fig:stateTubeExample}
	\end{subfigure}
	\caption{Depiction of parameter set, predicted parameter set, state tube and predicted state tube}
\end{figure*}
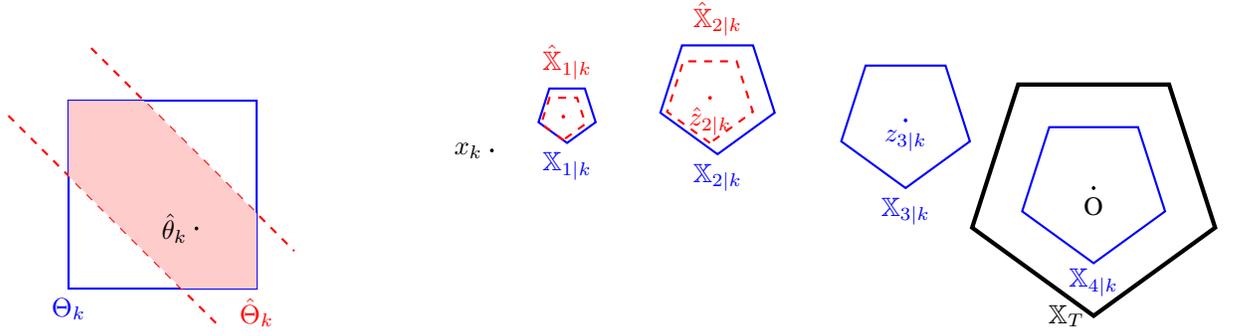
}
\subsection{Predicted state tube}
Using the set $ \hat{\Theta}_k $, a predicted state tube of length $ \hat{N} \le N $ is constructed such that $ \forall \theta \in \hat{\Theta}_k$ and $  l = \mathbb{N}_{0}^{\hat{N}-1}$
\begin{align}\label{eq:PredictedStateTube}
\begin{split}
\hat{\mathbb{X}}_{0|k} \ni \{x_{k}\}, \quad 
\hat{\mathbb{X}}_{l+1|k} \supseteq A(\theta)\hat{\mathbb{X}}_{l|k}  \oplus B(\theta)u_{l|k} \oplus \mathbb{W}.
\end{split}
\end{align}
The control input applied is the same for the set-dynamics \eqref{eq:SetDynamics} and \eqref{eq:PredictedStateTube}, while the parameter sets used are different. The evolution of the predicted state tube and the robust state tube is shown in Figure \ref{fig:stateTubeExample}. Since the parameter sets satisfy $ \hat{\Theta}_k \subseteq \Theta_k $, the predicted state tube lies within the robust state tube $ \left(\hat{\mathbb{X}}_{l|k} \subseteq \mathbb{X}_{l|k}\right) $. Each set in the predicted state tube is parameterized as $ \hat{\mathbb{X}}_{l|k} =\{\hat{z}_{l|k}\} \oplus \hat{\alpha}_{l|k}\mathbb{X}_0 $ where $ \hat{z}_{l|k} \in \mathbb{R}^n, \hat{\alpha}_{l|k}\in\mathbb{R}_{\ge0}, $ for $  l=\mathbb{N}_{0}^{\hat{N}}  $ are decision variables in the MPC optimization problem. To reformulate the predicted set-dynamics in \eqref{eq:PredictedStateTube}, the following definitions are used for $  j \in \mathbb{N}_{1}^{v} ,l \in \mathbb{N}_{0}^{\hat{N}-1} $
\setlength\arraycolsep{0pt}
\begin{equation}\label{eq:x_til_jlk}
\begin{array}{rll}
\hat{x}_{l|k}^{j} &= \hat{z}_{l|k} + \hat{\alpha}_{l|k} x^{j}, \quad &\hat{d}_{l|k}^{j} = A_0 x_{l|k}^{j} + B_0 u_{l|k}^{j} - \hat{z}_{l+1|k},  \\
\hat{u}_{l|k}^{j} &= K\hat{x}_{l|k}^{j} + v_{l|k}, \quad & \hat{D}_{l|k}^{j} = D(\hat{x}_{l|k}^{j},\hat{u}_{l|k}^{j}).
\end{array}
\end{equation}
The next proposition formulates the dynamics of the predicted state tube as constraints. The proof is similar to Proposition \nolinebreak \ref{Pr:SetDynamics} and is omitted.
\begin{proposition}\label{Pr:PredictedSetDynamics}
The predicted state tube $ \{\hat{\mathbb{X}}_{l|k}\}_{l=0}^{\hat{N}-1}$ satisfies the set-dynamics \eqref{eq:PredictedStateTube} if and only if for all $ j\in \mathbb{N}_{1}^{v} $ and $ l\in \mathbb{N}_{0}^{\hat{N}-1}$ there exists $ \hat{\Lambda}_{l|k}^{j} \in \mathbb{R}^{n_x\times (n_\theta+n_w)}_{\ge0}$ such that 
\begin{subequations}\label{eq:lambda_tilConstraints}
\begin{align}
-H_x \hat{z}_{0|k} - \hat{\alpha}_{0|k}\mathbf{1} &\le -H_x x_k, \\
\hat{\Lambda}_{l|k}^{j} \hat{h}_{\theta_k} + H_x \hat{d}_{l|k}^{j} -\hat{\alpha}_{l+1|k} \mathbf{1} &\le -\bar{w},\\
H_x \hat{D}_{l|k}^{j} &= \hat{\Lambda}_{l|k}^{j} \hat{H}_{\theta}.
\end{align}
\end{subequations} 
The constraints are bilinear in the variables since $ \hat{H}_{\theta}, \hat{h}_{\theta_k} $ are linearly dependent on the control input $ u_k $ as seen in \eqref{eq:PredictedParameterSet}.
\end{proposition}
\subsection{Predicted worst-case cost}
The cost function to be minimized is
\begin{align}\label{eq:MPCcost}
\begin{split}
J(\mathbf{v},\hat{N},N) &= \sum_{i=0}^{\hat{N}} l(\hat{\mathbb{X}}_{i|k},v_{i|k}) +\sum_{i=\hat{N}+1}^{N} l(\mathbb{X}_{i|k},v_{i|k}),
\end{split}
\end{align}
where $l(\mathbb{X},v) = \max_{x\in\mathbb{X}} ||Qx||_{\infty} + ||R(Kx+v)||_{\infty},$, and $ Q, R $ are positive definite matrices. A linear cost is chosen so that it can be reformulated using linear inequalities. The cost function is the sum of the predicted worst-case cost over the horizon $ \hat{N} $ and the worst-case cost over the remaining prediction horizon. This combination is used because propagating the predicted state tube is computationally expensive due to the bilinear constraints \eqref{eq:lambda_tilConstraints}, and results in a non-convex optimization problem. The parameter $ \hat{N} $ offers a trade-off between the computational complexity and active exploration. A higher value of $ \hat{N} $ increases the effect of a smaller parameter set $ \hat{\Theta}_k $ and thus promotes exploration. However, it also increases the number of bilinear constraints and the computational complexity. This trade-off will be exemplified in the numerical tests shown in Section \ref{sec:NumericalExample}.

\subsection{MPC algorithm}
The MPC optimization can now be defined using all the elements described above. The decision variables are 
\begin{equation}\label{eq:decisionVariables}
q_k = \left\{ 
\begin{array}{l}
\bigl\{z_{l|k},\alpha_{l|k},\{\Lambda^{j}_{l|k}\}_{j=1}^{v} \bigr\}_{l=0}^{N}, \{v_{l|k}\}_{l=0}^{N-1} ,\\
\bigl\{\hat{z}_{l|k},\hat{\alpha}_{l|k},\{\hat{\Lambda}^{j}_{l|k}\}_{j=1}^{v} \bigr\}_{l=0}^{\hat{N}}
\end{array} \right\},
\end{equation}
and the optimization problem can be written as
\setlength\arraycolsep{2pt}
\begin{equation}\label{eq:OptimizationProblem}
\begin{array}{r l}
\text{minimize}& J \eqref{eq:MPCcost} \\
\text{s.t.}& q_k \in  
\mathcal{Q}_k := \{q_k| \eqref{eq:lambdaConstraints},\eqref{eq:TerminalConstraint},\eqref{eq:lambda_tilConstraints}\},
\end{array}
\end{equation} 
where $ \mathcal{Q}_k $ represents the feasible region at time step $ k $.
\begin{algorithm} [t]
	\caption{Adaptive MPC with active exploration}\label{Alg:AMPC} 
	\begin{algorithmic}[1]
		\Statex \textbf{Offline} Choose $ K $, $ \bar{\alpha} $ and $ \mathbb{X}_0 $. Initialize $ h_{\theta_k} $ and $ \hat{\theta}_k $.
		\Statex \textbf{Online}
		\State $ k\gets 1 $
		\Repeat 
		\State Obtain the measurement $ x_k $ 
		\State Construct $ \Delta_k $ according to \eqref{eq:SimpleNonfalsified}
		\State Update $ h_{\theta_k} $ using \eqref{eq:Theta_k_LP} and compute $ \hat{\theta}_k $
		\State Solve optimization problem \eqref{eq:OptimizationProblem}
		\State Apply the control input  $ u_k = K x_k + v_{0|k}$
		\State $ k \gets k+1 $
		\Until 
	\end{algorithmic}
\end{algorithm}
The adaptive MPC algorithm with active exploration is described in Algorithm \ref{Alg:AMPC}. The values of the prestabilizing gain $K$, the state tube shape $ \mathbb{X}_0 $, terminal set bound $ \bar{\alpha}$ must be computed offline. The value of $h_{\theta_k} $ is initialized according to \eqref{eq:ParameterBounds}, and an initial guess is used for $ \hat{\theta}_k $. The following proposition establishes the control theoretic properties of the algorithm
\begin{proposition}
Let the assumptions \ref{As:Disturbance},\ref{As:Feedback} and \ref{As:TerminalSet} be satisfied and an initial feasible solution exist, that is, $ \mathcal{Q}_0 \neq \{\varnothing\} $. Then, the closed loop system using Algorithm \ref{Alg:AMPC} satisfies the following properties for all $ k > 0 $:
\begin{enumerate}[(i)]
	\item $ \theta ^* \in \Theta_k$
	\item $ \mathcal{Q}_k \neq \{\varnothing\} $
	\item $ (x_k,u_k) \in \mathbb{Z} $.
\end{enumerate}
\end{proposition}
\begin{proof}
The properties (i) and (ii) are proven by induction. Let $ \theta ^* \in \Theta_k$ for some $ k\ge0 $ which implies  $ \theta ^* \in \Delta_k$ according to \eqref{eq:Theta_k_update}. Since $ w_k \in \mathbb{W} $, the definition of the non-falsified set $ \Delta_{k+1} $ in  \eqref{eq:SimpleNonfalsified} implies $ \theta ^* \in \Delta_{k+1}$. This is because $ \Delta_{k+1} $ is constructed using the $ s-1 $ measurements used in $ \Delta_k $. Applying \eqref{eq:Theta_k_update} for the time step $ k+1 $ proves $ \theta ^* \in \Theta_{k+1}$.

Property (ii) implies recursive feasibility, i.e., if the MPC problem is feasible at the first time step, it remains feasible. For the proof, assume there exists a feasible solution at time step $ k \ge 0$. It is sufficient to find a feasible solution for the control variables $ \{v_{l|k}\}_{l=0}^{N-1} $ and state tube variables $  \{z_{l|k},\alpha_{l|k},\{\Lambda^{j}_{l|k}\}_{j=1}^{v} \}_{l=0}^{N-1} $ at the next time step to prove that the optimization problem is feasible. This is because the state tube satisfies the set-dynamics of the predicted state tube \eqref{eq:PredictedStateTube}. Consider the state tube $ \{\mathbb{X}_{l|k}\}_{l=0}^{N} $ computed at time step $ k $. Assumption \ref{As:TerminalSet} implies that the feedback controller $ u = Kx $ maps the set $ \mathbb{X}_{N|k} $ to a set $ \alpha^+ \mathbb{X}_0 $, where $\alpha^+ \in [0, \bar{\alpha}]$. Since the relation $ \Theta_k \supseteq \Theta_{k+1} $ holds according to \eqref{eq:Theta_k_update},  a feasible sequence of control inputs at the next time step is 
 \begin{equation*}
	\{v_{l|k+1}\}_{l=0}^{N-2} = \{v_{l+1|k}\}_{l=0}^{N-2},\quad v_{N-1|k+1} = 0.
 \end{equation*}
and a feasible sequence of sets defining the state tube is  
  \begin{align*}
 \{\mathbb{X}_{l|k+1}\}_{l=0}^{N-1} = \{\mathbb{X}_{l+1|k}\}_{l=0}^{N-1}, \quad \mathbb{X}_{N|k+1} = \alpha^+ \mathbb{X}_0.
 \end{align*}
 
 Property (iii) is a direct result of Proposition \ref{Pr:SetDynamics} and recursive feasibility.
\end{proof}

\section{Numerical results} \label{sec:NumericalExample}
In this section, the performance of the dual adaptive MPC (DAMPC) algorithm presented in this paper is compared to a passive adaptive MPC (PAMPC) algorithm from \cite{lorenzen2017}. The DAMPC algorithm performs active exploration using a predicted state tube, and the dependence of the performance its length $\hat{N}$ is studied. The PAMPC algorithm  uses the MPC cost function with $\hat{N}$ set to 0. The system matrices used in the simulation are given by
\begin{equation*}
\begin{array}{l l l}
A_0 = \begin{bmatrix} 0.85 &  0.5 \\ 0.2 & 0.6 \end{bmatrix}, & A_1 = \begin{bmatrix} 0.1 &  0 \\ 0 & 0.1 \end{bmatrix}, &A_2 = \begin{bmatrix} 0 &  0 \\ 0 & 0 \end{bmatrix}, \\
B_0 = \begin{bmatrix} 1 & 0.4\\ 0.2& 0.4 \end{bmatrix},  & B_1 = \begin{bmatrix} 0 &  0 \\ 0 & 0 \end{bmatrix},  & B_2 = \begin{bmatrix} 0 & 0.5\\ 0 & 0.4 \end{bmatrix}.  \\
\end{array} 
\end{equation*}
The uncertainty in the parameters is described by $ \Theta {=} \left\{\theta\in \mathbb{R}^2\bigr|\: ||\theta||_\infty \le 1 \right\}, $ with $ \theta^* = [0.95, 0.3]^\intercal $. The disturbance set is $ \mathbb{W} = \left\{w\in \mathbb{R}^2 \bigr|\: ||w||_\infty \le 0.1 \right\} $ and the state and input constraints are described by 
\begin{align*}
	\mathbb{Z} &= \left\{(x,u){\in} \mathbb{R}^{2\times 2}\left|\: \begin{array}{rl}
	||x||_\infty &\le 10 \\
	-0.5\le [u]_1 \le 1, &-2\le [u]_2 \le 2 \\ 
	\end{array}\right.  \right\}.
\end{align*}
The initial state of the system is $ x_0 = [1,1.5]^\intercal $. In both PAMPC and DAMPC, the state tube is constructed by translating and scaling the set $ \mathbb{X}_0 = \left\{x\in \mathbb{R}^2 \bigr|\: ||x||_\infty \le 1 \right\} $. The bounded complexity update of $ \Theta_k $ is performed using $ n_\theta = 58$ hyperplanes which are initially chosen as outer bounds of the set $ \Theta $. The cost matrices are given as $ Q =R= \mathbb{I}_{2\times2} $, the prestabilizing gain used is 
\begin{equation*}
	K = \begin{bmatrix}
	-0.5625 & 0 \\ 0 &0
	\end{bmatrix},
\end{equation*}
and the corresponding terminal set bound $ \bar{\alpha} $ is 0.89. The prediction horizon chosen is $ N = 8 $ time steps for all the algorithms. Two different values of the $ \hat{N} $ are used, and the corresponding adaptive MPC schemes are referred to as DAMPC$_2 $ and DAMPC$_5 $ for $ \hat{N}=2,5 $ respectively. The DAMPC schemes are initialized at $ \hat{\theta}_0 = [0.5,0.5]^\intercal $. 

\begin{figure}
	\includegraphics[scale=0.9]{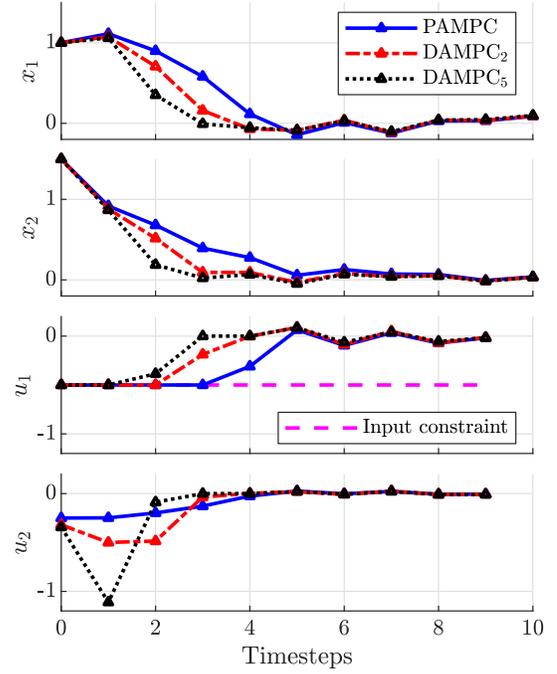}
	\centering
	\caption{Closed loop trajectories achieved under PAMPC and two DAMPC schemes with predicted state tube length $ \hat{N} = 2,5 $.}
	\label{fig:trajCL}
\end{figure}
\begin{figure}
\includegraphics[scale=0.9]{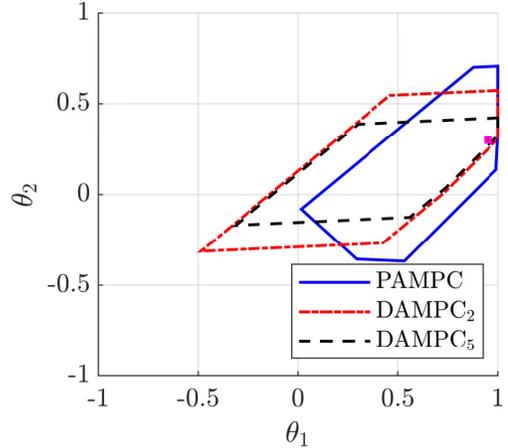}
\centering
\caption{Parameter sets obtained after running the adaptive MPC schemes for 10 timesteps. The square represents the true parameter. }
\label{fig:parameterSet}
\end{figure}
\begin{figure}
\includegraphics[scale=0.9]{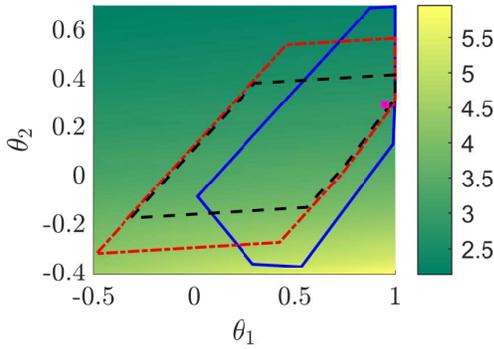}
\centering
\caption{Colormap showing the distribution of closed loop costs of MPC controllers as a function of the parameters. The legend is same as in Figure \ref{fig:parameterSet}, and is omitted for readability.}
\label{fig:NominalCost}
\end{figure}
The closed loop trajectories using each of the controllers are shown in Figure \ref{fig:trajCL}. It can be seen that active exploration improves the regulation performance. The PAMPC scheme achieves a closed loop cost of 6.04, while DAMPC$_2$ and DAMPC$_5$ achieve 4.49 (25\% lower) and 4.21 (30\% lower) respectively. The coefficients of control input $ [u]_2 $ have high uncertainty, and the PAMPC algorithm does not excite this input since the MPC optimizer within does not explicitly include the benefit of online identification. However, the DAMPC algorithms use a higher value of $ [u]_2 $ which improves the identification and reduces the closed loop cost. The updated uncertainty set $ \Theta_k $ of each scheme after 10 time steps is shown in Figure \ref{fig:parameterSet}. Even though the uncertainty sets for DAMPC$_2$ and PAMPC have similar size, the DAMPC$_2$ algorithm has a lower uncertainty in the parameter $ [\theta]_2 $ which has a stronger influence on the performance. This can be interpreted using Figure \ref{fig:NominalCost}, which shows the closed loop cost of an MPC controller specifically designed for each plant in the uncertainty set, plotted as a function of the corresponding parameters. Since the goal is to investigate the exploratory actions of the three aforementioned controllers, only the region around the uncertainty sets depicted in Figure \ref{fig:parameterSet} is considered. Figure \ref{fig:NominalCost} reveals the relative difficulty in controlling the systems, and thus motivates why some regions of parameter space are removed from the uncertainty set $ \Theta_k $ rather than the others. The figure shows that compared to the DAMPC$_2 $ controller, the PAMPC controller results in a parameter set associated with worse performance. This is because the predicted worst-case cost function induces exploration so as to remove the systems difficult to control from the future parameter set. Additionally, it can be seen that using a larger $ \hat{N} $ improves exploration. The DAMPC$_5$ algorithm has the smallest uncertainty set, while also having the least closed loop cost. The cost-reduction offered by DAMPC schemes is achieved at the price of computational complexity. The simulations were performed on a laptop using Intel i7-8550U 1.8 GHz processor, and the optimization problems were setup using YALMIP \cite{Lofberg2004} and solved using IPOPT \cite{wachter2006implementation}.  The average solver time for the optimization problem in PAMPC was 0.042s, while that of DAMPC$_2$ and DAMPC$_5$ were 0.89s and 1s respectively. A similar trend was observed for the performance and solver times with different values of the predicted state tube length. 
\section{Conclusion}
A dual adaptive MPC scheme was presented for systems with parametric uncertainty in state-space matrices. The algorithm uses online set-membership identification to reduce the uncertainty in the parameters and a tube MPC approach to ensure robust constraint satisfaction. A predicted state-tube is used to capture the effect of the future control inputs on identification, and a predicted worst-case cost is optimized. The resulting optimization problem in the MPC is non-convex, but offers the flexibility to trade-off the computational complexity with performance. The algorithm ensures recursive feasibility and consistency of the parameter set, and performs better compared to a passive adaptive MPC approach from literature while regulating a system. 
\bibliographystyle{ieeetr}
\bibliography{bibliography}             

\end{document}